\documentclass[conference]{IEEEtran}
\IEEEoverridecommandlockouts

\usepackage{cite}
\usepackage{amsmath,amssymb,amsfonts}
\usepackage{graphicx}
\usepackage{textcomp}
\usepackage{xcolor}
\def\BibTeX{{\rm B\kern-.05em{\sc i\kern-.025em b}\kern-.08em
    T\kern-.1667em\lower.7ex\hbox{E}\kern-.125emX}}

\usepackage{cite}
\usepackage{amsmath,amssymb,amsfonts,mathtools}
\usepackage{amsthm}
\usepackage{algorithm}
\usepackage{algpseudocode}
\usepackage{graphicx}
\usepackage{booktabs}
\usepackage{multirow}
\usepackage{array}
\usepackage{url}
\usepackage{xcolor}
\usepackage{tikz}
\usetikzlibrary{arrows.meta,positioning,fit,calc}

\newtheorem{definition}{Definition}
\newtheorem{assumption}{Assumption}
\newtheorem{proposition}{Proposition}
\newtheorem{lemma}{Lemma}

\DeclareMathOperator{\Cov}{Cov}
\DeclareMathOperator{\CVaR}{CVaR}

\DeclareMathOperator{\dist}{dist}

\newcommand{\R}{\mathbb{R}}
\newcommand{\E}{\mathbb{E}}
\newcommand{\cN}{\mathcal{N}}
\newcommand{\cX}{\mathcal{X}}
\newcommand{\cA}{\mathcal{A}}
\newcommand{\cK}{\mathcal{K}}
\newcommand{\cF}{\mathcal{F}}
\newcommand{\cT}{\mathcal{T}}

\newcommand{\bphi}{\boldsymbol{\Phi}}

\newcommand{\one}{\mathbf{1}}

\begin{document}

\title{ 
Differentiable Policy Transport over Multi-Layer Network Feasibility Geometry
}

\author{
\begin{tabular}{ccc}

\begin{tabular}[t]{c}
Zuyuan Zhang\\[-0.15em]
{\small\textit{The George Washington University, USA}}\\[-0.15em]
{\small zuyuan.zhang@gwu.edu}
\end{tabular}

&

\begin{tabular}[t]{c}
Zeyu Fang\\[-0.15em]
{\small\textit{The George Washington University, USA}}\\[-0.15em]
{\small joey.fang@gwu.edu}
\end{tabular}

&

\begin{tabular}[t]{c}
Mahdi Imani\\[-0.15em]
{\small\textit{Northeastern University, USA}}\\[-0.15em]
{\small m.imani@northeastern.edu}
\end{tabular}

\\[3.0em]

\multicolumn{3}{c}{
\begin{tabular}{cc}

\begin{tabular}[t]{c}
Nathaniel D. Bastian\\[-0.15em]
{\small\textit{Johns Hopkins University, USA}}\\[-0.15em]
{\small\textit{Syracuse University, USA}}\\[-0.15em]
{\small ndbastian@jhu.edu}
\end{tabular}

& \hspace{3.5em}

\begin{tabular}[t]{c}
Tian Lan\\[-0.15em]
{\small\textit{The George Washington University, USA}}\\[-0.15em]
{\small tlan@gwu.edu}
\end{tabular}

\end{tabular}
}

\end{tabular}
}

\maketitle

\begin{abstract}
Learning-based control is increasingly central to automating network operations. 
A learned policy, however, must satisfy cross-layer constraints on interference, power-rate coupling, flow conservation, service chains, capacity, latency, and reliability.
Existing methods typically account for only a subset of this geometry and only indirectly, e.g., through reward penalties, Lagrange multipliers, or post-hoc repairs. This paper proposes \emph{Network Feasibility Geometry Reinforcement Learning} (NFG-RL), which models coupled constraints via transport theory and the residual inclusion $\bphi_{\mathfrak{N}}(x,a)\in\cK_{\mathfrak{N}}$, defining the executed policy as the pushforward of a proto-policy through a feasibility-transport map.
NFG-RL compiles heterogeneous constraints into typed residual blocks and transports proto-actions through a differentiable variational operator, letting active constraints shape execution, exploration, and actor gradients.
Our analysis shows that exact transport yields almost-sure feasible execution, while active constraints contract exploration onto the feasible tangent space. It further establishes a nonnegative first-order gain from critic-tilted transport over plain projection and recovers backpressure scheduling as the gradient of a lifted drift residual. In two public-trace-conditioned wireless-edge surrogate environments,
NFG-RL improves feasible utility by
\textbf{37.5--41.5\%}
over the strongest non-NFG method in each environment,
reduces raw-action violation by
\textbf{48.5--60.8\%},
and lowers P99 delay by
\textbf{57.0--75.5\%}, outperforming a range of optimization and learning baselines. 
\end{abstract}


\section{Introduction}

Modern communication networks increasingly rely on learning-based control across multiple network layers, encompassing wireless resource allocation, congestion control, management, routing, scheduling, admission, and service placement ~\cite{mao2017survey,mach2017mobile,xu2018experience,liu2021drl,zhang2024distributed,zhang2025network,zhang2026counterfactual,zhang2026lisfc}.
To this end, Markov Decision Process (MDP) models and Reinforcement Learning are particularly attractive, because neural policies can adapt to complex network dynamics, such as stochastic traffic, time-varying channels, failures, and elastic service objectives, that are difficult to model accurately or completely~\cite{luong2019applications,xu2018experience,zheng2021leveraging,zou2024distributed,qiao2024br}.

A learned policy must satisfy the physical laws, structural dependencies, and operational limits imposed by different network layers, which together determine whether a control action is executable. We refer to this coupled structure as the network feasibility geometry. 
Examples include flow conservation and service-chain continuity that require topological and logical consistency; link, spectrum, CPU, GPU, memory, and buffer limits that impose resource feasibility; wireless interference and power-rate coupling that mandate physical compatibility; queue stability and virtual budgets that cause temporal dependencies; and slicing, reliability, and tail-latency requirements that impose service-level restrictions~\cite{tassiulas1990stability,neely2010stochastic,chen2021bringing,rockafellar2000optimization}. 
Network feasibility geometry unifies the representation of these physical laws, dependencies, and limits to define the set of actions that the network can execute~\cite{tassiulas1990stability,neely2010stochastic,mach2017mobile,zhang2026geometry,zhang2026metric,zhang2026geometry}. 
Strict admissibility based on this network geometry must be settled before an action reaches the network.

\begin{table*}[t]
\centering
\caption{ Network law primitives represented by typed residual inclusions
$\bphi^{\ell}(x,a)\in\cK^{\ell}$; residual dimensions may differ across blocks.}
\label{tab:primitives}
\footnotesize
\begin{tabular}{p{0.15\textwidth}p{0.27\textwidth}p{0.22\textwidth}p{0.24\textwidth}}
\toprule
Primitive & Residual/set & Action derivative & Examples \\
\midrule
Balance 
& $\bphi^{\mathrm{bal}}(x,a)=B_{\mathfrak{N}}a-b(x)\in\{0\}$ 
& $D_a\bphi^{\mathrm{bal}}=B_{\mathfrak{N}}$ 
& Flow conservation, traffic conservation, service-stage continuity \\

Capacity 
& $\bphi^{\mathrm{cap}}(x,a)=R_{\mathfrak{N}}a-C_{\mathfrak{N}}(x,a)\in\R_-^m$ 
& $R_{\mathfrak{N}}-D_aC_{\mathfrak{N}}(x,a)$ 
& Link capacity, spectrum, CPU, GPU, memory, buffer limits \\

Physical coupling 
& $\bphi^{\mathrm{phy}}(x,a)=H_{\mathfrak{N}}(x,a)\in\R_-^m$ 
& $D_aH_{\mathfrak{N}}(x,a)$ 
& Wireless interference, SINR, power-rate coupling, half-duplex constraints \\

Dynamic stability 
& $\bphi^{\mathrm{dyn}}(x,a)=\E[V(y^+)-V(y)\mid x,a]-\delta_{\mathrm{dyn}}(x)\in\R_-$ 
& Expected drift gradient or sampled drift gradient 
& Queue stability, congestion stability, virtual queues, long-term budgets \\

Service logic
& $\bphi^{\mathrm{svc}}(x,a)=H_{\mathrm{svc}}(x,a)\in\{0\}$
& $D_aH_{\mathrm{svc}}(x,a)$ after relaxation
& Service chaining, placement dependency, affinity constraints \\

Risk and SLA
& $\bphi^{\mathrm{risk}}(x,a)=H_{\mathrm{risk}}(x,a)\in\R_-^m$
& $D_aH_{\mathrm{risk}}(x,a)$ or a subgradient
& Tail latency, reliability, loss, deadline violation \\

Discrete feasibility 
& $\bphi^{\mathrm{disc}}(x,a) =G_{\mathrm{disc}}(x)a_{\mathrm{disc}}-h_{\mathrm{disc}}(x)\in\R_-^r$ 
& $G_{\mathrm{disc}}(x)$ or differentiable relaxation 
& Path selection, channel assignment, matching, placement, admission \\
\bottomrule
\end{tabular}
\end{table*}

Learning-based approaches often account for \emph{only a subset} of this geometry and \emph{only indirectly}. Penalty-based RL models violation as a penalty on attainable reward, primal--dual and constrained RL recast constraints into multipliers, and projection or repair methods correct an infeasible action after it has already been sampled~\cite{altman2021constrained,achiam2017constrained,chow2018lyapunov,tessler2018reward,amos2017optnet,agrawal2019differentiable}. Another line of work encodes structure into the policy input or architecture, using graph neural networks for topology awareness and attention for logical dependencies~\cite{geyer2019deeptma,rusek2020routenet,almasan2022deep}. They often provide soft enforcement through loss minimization. Closest to our setting, action-constrained and safety-layer methods wrap the actor in a projection, generative, or reduced-gradient feasibility layer~\cite{dalal2018safe,pham2018optlayer,donti2021dc3,lin2021escaping,brahmanage2023flowpg,liang2023low,tabas2022computationally,ding2023reduced}, with safety-aware training~\cite{xiao2023safe} and applications to network slicing~\cite{liu2021constraintaware,liu2021onslicing,nagib2025safeslice}. They typically focus on individual layers or dimensions~\cite{lin2021escaping,kasaura2023benchmarking} of the network geometry and fail to enable a holistic solutions.

This paper proposes \emph{Network Feasibility Geometry Reinforcement Learning} (NFG-RL). 
We model a networked decision system by a residual inclusion $\bphi_{\mathfrak{N}}(x,a)\in\cK_{\mathfrak{N}}$, which induces the feasible set $\cF_{\mathfrak{N}}(x)=\{a\in\cA:\bphi_{\mathfrak{N}}(x,a)\in\cK_{\mathfrak{N}}\}$. A single differentiable transport then enforces every layer of this feasible set at once. 
In particular, NFG-RL samples a proto-action from an unconstrained proto-policy and transports it through a state-dependent feasibility map. 
The executed policy is then the pushforward distribution induced by this transport map, so feasibility shapes the action distribution before execution. The transport differs from generic projection layers in two ways: it is anisotropic in a metric built from the residual Jacobians, and it is tilted by the critic gradient, so the executed action is feasible and, to first order, improves the estimated value over plain projection (Proposition~\ref{prop:improvement}).
A local variational transport operator over the residual Jacobians exposes the tangent and normal cones of the active constraints, and its KKT system pulls residual-space normal forces back into action coordinates~\cite{boyd2004convex,rockafellar1998variational,amos2017optnet,agrawal2019differentiable,blondel2022efficient,donti2017task}, so the same constraint information shapes both the executed action and the actor gradient.

NFG-RL inserts this feasibility transport between a proto-actor and the environment; Table~\ref{tab:primitives} lists the constraint families that compile into the interface it enforces. One compiled triple $(\bphi_{\mathfrak{N}},\cK_{\mathfrak{N}},J_{\Phi})$ thereby acts on the policy in three places. It fixes where the executed policy may place probability mass, since exact transport yields almost-sure feasible execution along entire trajectories (Lemmas~\ref{lem:as_execution_feasibility} and~\ref{lem:transport_compatibility}). It contracts exploration noise onto the feasible tangent space at active constraints (Proposition~\ref{prop:covariance}). It filters the critic gradient through the transport Jacobian, with a provable first-order value gain over plain projection (Proposition~\ref{prop:improvement}), and a lifted drift residual recovers backpressure scheduling as a special case (Proposition~\ref{prop:backpressure}).

We implement NFG-RL in a multi-tenant wireless edge computing system and evaluate it in public-trace-conditioned wireless-edge environments, SMEC-5G and 5G-C3~\cite{zhang2026enabling,raca2020beyond,tocze2022edge,mehran2022matching}. The geometry we consider spans routing, wireless scheduling, power control, compute allocation, service placement, admission, queue stability, slicing, reliability, tail latency, and discrete execution decisions~\cite{mao2017survey,mach2017mobile,chen2021bringing}.
NFG-RL outperforms a range of baselines including classical heuristics, scalar penalty and
primal--dual baselines, a tenant-interaction graph policy, a generic
projection baseline, and mechanism and residual-family ablations
~\cite{schulman2017proximal,haarnoja2018soft,achiam2017constrained,
tassiulas1990stability,geyer2019deeptma}.

The paper makes the following contributions:
\begin{itemize}
    \item We formulate learning-based network control through one typed residual inclusion $\bphi_{\mathfrak{N}}(x,a)\in\cK_{\mathfrak{N}}$ and optimize over policies supported on its induced feasible sets.
    
    \item We introduce anisotropic, critic-tilted feasibility transport and its pushforward policy. Zero tilt fixes feasible actions, active constraints contract exploration onto the tangent space, the tilt gives a nonnegative first-order gain over plain projection, and a lifted drift residual recovers backpressure.
    
    \item We compile fourteen constraint families of a multi-tenant wireless edge system into this interface and evaluate NFG-RL against classical, penalty, primal--dual, graph, and projection baselines in public-trace-conditioned environments.
\end{itemize}

\section{Networked Decision Systems}
\label{sec:nds}

We first fix the decision object and notation for the geometric construction in Section~\ref{sec:nfg}. A network controller is modeled as a Markov decision system whose constraint map compiles heterogeneous constraints into the residual inclusion $\bphi_{\mathfrak{N}}(x,a)\in\cK_{\mathfrak{N}}$; admissibility is then required of the support of the executed policy.  Concrete queues, links, channels, routes, and servers enter only through the coordinates of $(x,a)$ and the residual blocks of this inclusion.

\begin{definition}[Networked decision system]
\label{def:nds}
A networked decision system is a tuple $\mathfrak{N}=(\cX,\cA,P,U,\bphi,\cK)$, where $\cX$ is a standard Borel state space, $\cA$ is a nonempty Borel subset of a finite-dimensional Euclidean space $\mathbb{A}$, $P(dx'\mid x,a)$ is a Markov kernel on $\cX$, and $U:\cX\times\cA\to\R$ is measurable.  The measurable law map $\bphi:\cX\times\cA\to\mathcal{Z}_{\mathfrak{N}}$ takes values in a finite-dimensional normed residual space, and $\cK\subseteq\mathcal{Z}_{\mathfrak{N}}$ is nonempty and closed. 
We assume that $U$ is bounded or that the discounted return is integrable under the policies considered; when emphasizing the system, we write $\bphi_{\mathfrak{N}}$ and $\cK_{\mathfrak{N}}$, and write $\mathcal{Z}$ for $\mathcal{Z}_{\mathfrak{N}}$ when no ambiguity arises.
\end{definition}

Here $x$ is the Markov information available to the controller, such as a physical state, augmented history, belief, or sufficient statistic, and $a$ is the executed routing, scheduling, power, compute, admission, placement, or path decision.  The inclusion $\bphi_{\mathfrak{N}}(x,a)\in\cK_{\mathfrak{N}}$ therefore decides admissibility before execution.

\begin{assumption}[Basic regularity]
\label{ass:basic_regularity}
For the geometric arguments below, $\cA$ is closed in $\mathbb{A}$, possibly after replacing a discrete action set by a closed continuous relaxation; $\cK$ is closed; and $\bphi$ is jointly measurable and continuous in $a$ for each fixed $x$.  Where first-order information is used, $\bphi(x,\cdot)$ is differentiable; for locally Lipschitz nonsmooth primitives, $D_a\bphi(x,a)$ denotes a fixed generalized-Jacobian element selected by the transport solver.  The feasible set is nonempty at every modeled state, possibly because $\cA$ contains a drop, no-op, safety, or repair action.
\end{assumption}

No convexity or global smoothness is assumed: continuity is used only for closedness, and first-order regularity only where the transport operator evaluates $D_a\bphi(x,a)$.

\begin{definition}[Network admissibility]
For $x\in\cX$, an action $a\in\cA$ is network-admissible if $\bphi_{\mathfrak{N}}(x,a)\in\cK_{\mathfrak{N}}$, and the feasible action set is
\begin{equation}
\label{eq:feasible_set}
    \cF_{\mathfrak{N}}(x)
    =
    \{a\in\cA:\bphi_{\mathfrak{N}}(x,a)\in\cK_{\mathfrak{N}}\}.
\end{equation}
For $\varepsilon\ge0$, define the relaxed set $\cF_{\mathfrak{N}}^{\varepsilon}(x)=\{a\in\cA:\dist_{\mathcal{Z}_{\mathfrak{N}}}(\bphi_{\mathfrak{N}}(x,a),\cK_{\mathfrak{N}})\le\varepsilon\}$; closedness of $\cK_{\mathfrak{N}}$ gives $\cF_{\mathfrak{N}}^{0}(x)=\cF_{\mathfrak{N}}(x)$.
\end{definition}

The next lemma supplies the closedness and measurable-graph properties required for state-dependent action constraints.

\begin{lemma}[Well-posed feasible sets]
\label{lem:well_posed_feasible_sets}
Under Assumption~\ref{ass:basic_regularity}, $\cF_{\mathfrak{N}}(x)$ and $\cF_{\mathfrak{N}}^{\varepsilon}(x)$ are closed for every $x$ and $\varepsilon\ge0$, and their graphs are measurable subsets of $\cX\times\cA$.
\end{lemma}

\begin{proof}[Proof sketch]
Let $d_{\mathfrak{N}}(x,a)=\dist_{\mathcal{Z}_{\mathfrak{N}}}(\bphi_{\mathfrak{N}}(x,a),\cK_{\mathfrak{N}})$. 
For fixed $x$, continuity of $\bphi_{\mathfrak{N}}(x,\cdot)$ and of distance to a closed set makes $d_{\mathfrak{N}}(x,\cdot)$ continuous, so $\cF_{\mathfrak{N}}^{\varepsilon}(x)=\{a\in\cA:d_{\mathfrak{N}}(x,a)\le\varepsilon\}$ is closed and $\cF_{\mathfrak{N}}(x)=\cF_{\mathfrak{N}}^{0}(x)$. 
Joint measurability of $\bphi_{\mathfrak{N}}$ makes $d_{\mathfrak{N}}$ measurable; hence $\operatorname{Gr}(\cF_{\mathfrak{N}}^{\varepsilon})=\{(x,a)\in\cX\times\cA:d_{\mathfrak{N}}(x,a)\le\varepsilon\}$ is measurable, including the exact case $\varepsilon=0$.
\end{proof}

Equality constraints use $\bphi=h$ and $\cK=\{0\}$, inequalities $g(x,a)\le0$ use $\bphi=g$ and $\cK=\R_-^m$, and conic constraints choose the corresponding cone. 
Discrete constraints may use a continuous relaxation during learning, but exact feasibility applies only after rounding or repair returns the executed action to $\cF_{\mathfrak{N}}(x)$; otherwise the policy is only $\varepsilon$-compatible. 
Long-term constraints may augment $x$ with a virtual state and include its drift as a residual block, but a stability or drift argument is still required for a long-term guarantee~\cite{neely2010stochastic}.

We now lift action feasibility to policy support: the executed policy must assign probability one to the state-dependent feasible set rather than merely penalizing violations.

\begin{definition}[Network-compatible policy]
A Markov policy $\pi(da\mid x)$ is exactly network-compatible with $\mathfrak{N}$ if $\pi(\cF_{\mathfrak{N}}(x)\mid x)=1$ for every $x\in\cX$; let $\Pi_{\mathfrak{N}}$ denote all such policies. 
For $\varepsilon\ge0$, it is $\varepsilon$-network-compatible if $\pi(\cF_{\mathfrak{N}}^{\varepsilon}(x)\mid x)=1$ for every $x$, and $\Pi_{\mathfrak{N}}^{\varepsilon}$ denotes this relaxed class.
\end{definition}

Thus a policy may randomize within $\cF_{\mathfrak{N}}(x)$, while $\Pi_{\mathfrak{N}}^{\varepsilon}$ models approximate transport or training tolerances before optional repair.

\begin{lemma}[Almost-sure execution feasibility]
\label{lem:as_execution_feasibility}
Along trajectories generated by $(\mu_0,\pi,P)$, every $\pi\in\Pi_{\mathfrak{N}}$ satisfies $\bphi_{\mathfrak{N}}(x_t,a_t)\in\cK_{\mathfrak{N}}$ for all $t\ge0$ almost surely. 
Every $\pi\in\Pi_{\mathfrak{N}}^{\varepsilon}$ similarly satisfies $\dist_{\mathcal{Z}_{\mathfrak{N}}}(\bphi_{\mathfrak{N}}(x_t,a_t),\cK_{\mathfrak{N}})\le\varepsilon$ for all $t\ge0$ almost surely.
\end{lemma}

\begin{proof}[Proof sketch]
For each fixed $t$, the support condition makes the corresponding feasibility event have probability one. 
The intersection over the countable index set $t\ge0$ still has probability one; replacing $\cF_{\mathfrak{N}}$ by $\cF_{\mathfrak{N}}^{\varepsilon}$ proves the relaxed statement.
\end{proof}

Given an initial distribution $\mu_0$ on $\cX$ and $\gamma\in[0,1)$, the exact network-constrained RL problem is
\begin{equation}
\label{eq:nfg_objective}
    \max_{\pi\in\Pi_{\mathfrak{N}}}
    J_{\mathfrak{N}}(\pi)
    =
    \E_{\mu_0,\pi,P}
    \left[
        \sum_{t=0}^{\infty}\gamma^t U(x_t,a_t)
    \right].
\end{equation}
By Lemma~\ref{lem:as_execution_feasibility}, the policy constraint enforces almost-sure feasible execution at every time.

Approximate implementations either optimize over $\Pi_{\mathfrak{N}}^{\varepsilon}$, with residual distance quantifying the remaining error, or repair before execution; only the repaired action supports an exact-feasibility claim.  NFG therefore places feasibility in the support of the executed policy itself; reward penalties, Lagrange costs, and post-hoc violation scores act only on realized violations and leave that support unconstrained.

\section{Network Feasibility Geometry}
\label{sec:nfg}

This section turns $\bphi_{\mathfrak{N}}(x,a)\in\cK_{\mathfrak{N}}$ into a differentiable action-generation interface, deriving its local geometry, pushforward policy, and KKT-filtered actor gradients.

\subsection{Local Tangent and Normal Geometry}

For feasible $(x,a)$, let $J_{\Phi}(x,a)=D_a\bphi_{\mathfrak{N}}(x,a)$. For a closed set $S$, $T_S$ and $N_S$ denote Bouligand tangents and limiting normals; for convex $S$, these are the standard convex-analytic cones~\cite{rockafellar1998variational}. Nonsmooth primitives use the generalized Jacobian selected in Assumption~\ref{ass:basic_regularity}.

Define the feasible first-order directions and pulled-back normals by
\begin{align}
\label{eq:tangent}
\cT_{\mathfrak{N}}(x,a)
&=
\left\{
v\in T_{\cA}(a):
J_{\Phi}(x,a)v
\in
T_{\cK_{\mathfrak{N}}}
\bigl(\bphi_{\mathfrak{N}}(x,a)\bigr)
\right\},
\\
\label{eq:normal}
\cN_{\mathfrak{N}}(x,a)
&=
N_{\cA}(a)
+
J_{\Phi}(x,a)^\top
N_{\cK_{\mathfrak{N}}}
\bigl(\bphi_{\mathfrak{N}}(x,a)\bigr).
\end{align}
The first set preserves feasibility to first order; the second pulls residual normals into action coordinates through $J_{\Phi}^{\top}$.

Exact identities require the following local regularity.

\begin{assumption}[Local cone regularity]
\label{ass:local_cone_calculus}
Whenever exact cone identities are invoked, $\cF_{\mathfrak{N}}(x)=\cA\cap\bphi_{\mathfrak{N}}(x,\cdot)^{-1}(\cK_{\mathfrak{N}})$ satisfies the tangent intersection/chain and normal sum/pullback rules. It suffices that $\cA$ and $\cK_{\mathfrak{N}}$ are Clarke regular and the only $\lambda\in N_{\cK_{\mathfrak{N}}}(\bphi_{\mathfrak{N}}(x,a))$ satisfying $-J_{\Phi}^{\top}\lambda\in N_{\cA}(a)$ is $0$; for smooth inequalities this is MFCQ~\cite{rockafellar1998variational}. Otherwise, \eqref{eq:tangent}--\eqref{eq:normal} are first-order surrogates.
\end{assumption}

\begin{lemma}[Exact local feasibility geometry]
\label{lem:local_cone_identity}
Under Assumption~\ref{ass:local_cone_calculus}, every regular feasible pair
$(x,a)$ satisfies
$T_{\cF_{\mathfrak{N}}(x)}(a)=\cT_{\mathfrak{N}}(x,a)$ and
$N_{\cF_{\mathfrak{N}}(x)}(a)=\cN_{\mathfrak{N}}(x,a)$.
\end{lemma}

\begin{proof}[Proof sketch]
The tangent inverse-image/intersection rules and normal pullback/sum rules applied to $\cA\cap\bphi_{\mathfrak{N}}(x,\cdot)^{-1}(\cK_{\mathfrak{N}})$ give both identities.
\end{proof}

\subsection{Proto-Policy and Pushforward Policy}

NFG separates proposal from execution: $u\sim\rho_{\theta}(\cdot\mid x)$ is mapped by measurable $\Gamma_{\mathfrak{N},x}:\mathbb{A}\to\cA$. The map is exact if $\Gamma_{\mathfrak{N},x}(u)\in\cF_{\mathfrak{N}}(x)$ almost surely, and $\varepsilon$-feasible with $\cF_{\mathfrak{N}}^{\varepsilon}(x)$. The executed policy is
\begin{equation}
\label{eq:pushforward}
\pi_{\theta}^{\mathfrak{N}}(\cdot\mid x)=(\Gamma_{\mathfrak{N},x})_{\#}\rho_{\theta}(\cdot\mid x),
\end{equation}
i.e., $\pi_{\theta}^{\mathfrak{N}}(B\mid x)=\rho_{\theta}(\Gamma_{\mathfrak{N},x}^{-1}(B)\mid x)$.

\begin{lemma}[Exact transport implies network compatibility]
\label{lem:transport_compatibility}
Exact transport at every state implies
$\pi_{\theta}^{\mathfrak{N}}\in\Pi_{\mathfrak{N}}$;
$\varepsilon$-feasible transport implies
$\pi_{\theta}^{\mathfrak{N}}\in\Pi_{\mathfrak{N}}^{\varepsilon}$.
\end{lemma}

\begin{proof}[Proof sketch]
Pushforward assigns probability one to $\cF_{\mathfrak{N}}(x)$ (or $\cF_{\mathfrak{N}}^{\varepsilon}(x)$) because its preimage has proto-policy probability one.
\end{proof}

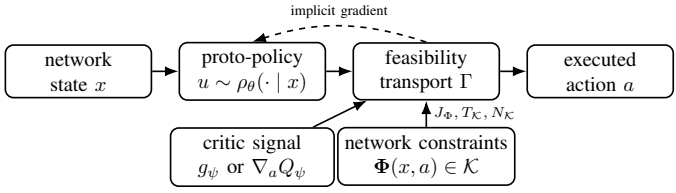
\begin{figure}[t]
\centering
\resizebox{\columnwidth}{!}{%
\begin{tikzpicture}[node distance=4.5mm,>=Latex,thick, every node/.style={transform shape}]
\tikzstyle{box}=[draw,rounded corners,align=center,minimum height=8mm,minimum width=25mm]
\node[box] (x) {network\\state $x$};
\node[box,right=of x] (proto) {proto-policy\\$u\sim\rho_\theta(\cdot\mid x)$};
\node[box,right=of proto] (trans) {feasibility\\transport $\Gamma$};
\node[box,right=of trans] (act) {executed\\action $a$};
\node[box,below=of trans,minimum width=28mm] (geom) {network constraints\\$\bphi(x,a)\in\cK$};
\node[box,below=of proto,minimum width=28mm] (critic) {critic signal\\$g_\psi$ or $\nabla_a Q_\psi$};
\draw[->] (x) -- (proto);
\draw[->] (proto) -- (trans);
\draw[->] (trans) -- (act);
\draw[->] (geom) -- node[right,font=\scriptsize] {$J_\Phi,T_\cK,N_\cK$} (trans);
\draw[->] (critic) -- (trans);
\draw[->,dashed] (trans.north) to[bend right=18] node[above,font=\scriptsize] {implicit gradient} (proto.north);
\end{tikzpicture}%
}
\caption{Proto-actions are transported to $\cF_{\mathfrak{N}}(x)$ and differentiated through the feasibility geometry.}
\label{fig:framework}
\end{figure}

\subsection{Feasibility Transport as a Variational Operator}

Given $u$, choose a reference $\bar a(x,u)$ in the domain of $Q_{\psi}$ and $\bphi_{\mathfrak{N}}$, and a direction $g_{\psi}(x,u)$, typically $\nabla_aQ_{\psi}(x,a)|_{a=\bar a(x,u)}$. NFG selects a nearby feasible action by
\begin{align}
\label{eq:transport}
\Gamma_{\mathfrak{N},x}(u)
&\in
\operatorname*{arg\,max}_{a\in\cF_{\mathfrak{N}}(x)}
\Psi_x(u,a),
\\
\label{eq:transport_objective}
\Psi_x(u,a)
&=
\langle g_{\psi}(x,u),a-u\rangle
-
\frac{1}{2\eta}
\|a-u\|^2_{M_{\mathfrak{N}}(x,u)},
\end{align}
where $\eta>0$, $M_{\mathfrak{N}}(x,u)\succ0$, and $\|v\|_M^2=\langle v,Mv\rangle$. Completing the square makes \eqref{eq:transport} the $M_{\mathfrak{N}}$-metric projection of $\tilde u_{\psi}=u+\eta M_{\mathfrak{N}}^{-1}g_{\psi}$ onto $\cF_{\mathfrak{N}}(x)$~\cite{boyd2004convex,amos2017optnet,agrawal2019differentiable}.

A canonical metric is
\begin{equation}
\label{eq:metric}
M_{\mathfrak{N}}(x,u)=I+\varsigma J_{\Phi}(x,\bar a)^{\top}W(x,u)J_{\Phi}(x,\bar a),
\end{equation}
where $W\succeq0$ weights residual sensitivities and $\varsigma\ge0$ controls anisotropy. Use $\bar a=u$ only when both maps are defined there; otherwise use a clipped, decoded, feasible, or solver-iterate reference. Thus $g_{\psi}=0$ gives network-aware projection, while a critic gradient gives local feasible improvement. Nonconvex or discrete relaxations require final rounding or repair for exact feasibility.

\begin{lemma}[Well-posed local transport]
\label{lem:well_posed_transport}
Fix $x,u$. If $\cF_{\mathfrak{N}}(x)$ is nonempty and closed, $M_{\mathfrak{N}}\succ0$, and $M_{\mathfrak{N}},g_{\psi}$ are independent of $a$, then \eqref{eq:transport} has a maximizer, unique for convex $\cF_{\mathfrak{N}}(x)$. Otherwise, the measurable-graph property and Carath\'eodory objective give a measurable maximizer selection~\cite[Ch.~14]{rockafellar1998variational}.
\end{lemma}

\begin{proof}[Proof sketch]
With $\tilde u_{\psi}$ defined above,
$\Psi_x(u,a)
=-(2\eta)^{-1}
\|a-\tilde u_{\psi}(x,u)\|_{M_{\mathfrak{N}}(x,u)}^2+\mathrm{const}$.
This continuous objective tends to $-\infty$ as $\|a\|\to\infty$, so it
attains a maximum on the nonempty closed feasible set; strict concavity
gives uniqueness on a convex feasible set.
\end{proof}

At a regular solution $a^\star=\Gamma_{\mathfrak{N},x}(u)$, stationarity gives
\begin{equation}
\label{eq:kkt_compact}
0\in \eta^{-1}M_{\mathfrak{N}}(x,u)(a^\star-u)-g_{\psi}(x,u)+N_{\cF_{\mathfrak{N}}(x)}(a^\star).
\end{equation}
By Lemma~\ref{lem:local_cone_identity}, there are $\nu\in N_{\cA}(a^\star)$ and $\lambda\in N_{\cK_{\mathfrak{N}}}(\bphi_{\mathfrak{N}}(x,a^\star))$ such that
$
0=\eta^{-1}M_{\mathfrak{N}}(a^\star-u)-g_{\psi}+\nu+J_{\Phi}(x,a^\star)^\top\lambda.
$
Thus $J_{\Phi}^{\top}\lambda$ pulls residual normals into action coordinates; $\lambda$ is nonnegative on active inequalities, zero on inactive ones, and unrestricted for equalities. With $g_{\psi}=0$, transport fixes every feasible $u$, so it excludes no network-compatible policy representable by the proto-policy family.

\subsection{Backward Gradient Through Transport}

Let $u=u_{\theta}(x,\epsilon)$, $\epsilon\sim p$, and $a=\Gamma_{\mathfrak{N},x}(u)$. For $\mathcal{J}_{\mathrm{act}}(\theta)=\E_{x\sim d,\epsilon}[Q_{\psi}(x,\Gamma_{\mathfrak{N},x}(u_{\theta}(x,\epsilon)))]$, differentiable points satisfy
\begin{equation}
\label{eq:transport_chain_rule}
\nabla_{\theta}Q_{\psi}(x,a)=\left(\frac{\partial u_{\theta}}{\partial\theta}\right)^\top(D_u\Gamma_{\mathfrak{N},x})^\top\nabla_aQ_{\psi}(x,a).
\end{equation}

Averaging this identity over $(x,\epsilon)$ yields $\nabla_{\theta}\mathcal{J}_{\mathrm{act}}(\theta)$. The sensitivity $D_u\Gamma_{\mathfrak{N},x}$ depends on the active constraints, $M_{\mathfrak{N}}$, and $J_{\Phi}$, and can be computed by implicit KKT differentiation, solver unrolling, or a differentiable
optimization layer%
~\cite{agrawal2019differentiable,blondel2022efficient,donti2017task}.
At active-set changes, generalized Jacobians, subgradients, or straight-through estimators approximate this geometry-filtered gradient.

\subsection{Effect on the Executed Action Distribution}

Assume a locally stable affine active set and locally constant $M=M_{\mathfrak{N}}(x,u)$ and $g_{\psi}$. If $J_A$ stacks active constraints, the local transport sensitivity is
\begin{equation}
\label{eq:projection}
P_T^M=I-M^{-1}J_A^\top(J_AM^{-1}J_A^\top)^\dagger J_A,
\end{equation}
the $M$-orthogonal projector onto $\{v:J_Av=0\}$. Hence $\Cov[a\mid x]\approx P_T^M\Sigma_{\theta}(P_T^M)^\top$; Proposition~\ref{prop:covariance} formalizes this contraction of active normal directions.

\section{Residual Constraint Primitives}
\label{sec:primitives}

Heterogeneous network constraints compile into typed residual blocks as follows. Their product defines the global inclusion $\bphi_{\mathfrak{N}}(x,a)\in\cK_{\mathfrak{N}}$, while the stacked Jacobian preserves blockwise information for feasibility transport and actor updates.

Let $\mathcal{L}$ be a finite set of selected primitive types.
For each $\ell\in\mathcal{L}$, let
$\bphi^{\ell}:\cX\times\cA\to\mathcal{Z}^{\ell}$ and
$\cK^{\ell}\subseteq\mathcal{Z}^{\ell}$ denote its residual map and closed admissible set.
Define
$
\mathcal{Z}_{\mathfrak{N}}
=\prod_{\ell\in\mathcal{L}}\mathcal{Z}^{\ell},
\qquad
\bphi_{\mathfrak{N}}(x,a)
=\bigl(\bphi^{\ell}(x,a)\bigr)_{\ell\in\mathcal{L}},
\qquad
\cK_{\mathfrak{N}}
=\prod_{\ell\in\mathcal{L}}\cK^{\ell}.
$
Then $\bphi_{\mathfrak{N}}(x,a)\in\cK_{\mathfrak{N}}$ iff
$\bphi^{\ell}(x,a)\in\cK^{\ell}$ for every $\ell$.
A block-scaled product norm prevents raw rate, delay, queue, reliability, and power units from being mixed in the residual distance.
If the blocks are differentiable in $a$, write
$J_{\Phi}^{\ell}(x,a)=D_a\bphi^{\ell}(x,a)$ and compile
$J_{\Phi}(x,a)=\operatorname{col}_{\ell\in\mathcal{L}}J_{\Phi}^{\ell}(x,a)$.

\begin{lemma}[Primitive composition]
\label{lem:primitive_composition}
Suppose each $\cK^{\ell}$ is closed and $\bphi^{\ell}(x,\cdot)$ is continuous. Then $\cF_{\mathfrak N}(x)=\{a\in\cA:\bphi^{\ell}(x,a)\in\cK^{\ell},\ \forall\ell\in\mathcal L\}$. At regular feasible $a$ satisfying Assumption~\ref{ass:local_cone_calculus},
\begin{align}
\cT_{\mathfrak N}(x,a)
&=\left\{v\in T_{\cA}(a):\begin{array}{l}
J_{\Phi}^{\ell}(x,a)v\in T_{\cK^{\ell}}(\bphi^{\ell}(x,a)),\\[-0.2ex]
\forall\ell\in\mathcal L
\end{array}\right\},
\label{eq:product_tangent}\\[-0.3ex]
\cN_{\mathfrak N}(x,a)
&=N_{\cA}(a)+\sum_{\ell\in\mathcal L}
(J_{\Phi}^{\ell}(x,a))^{\top}
N_{\cK^{\ell}}(\bphi^{\ell}(x,a)).
\label{eq:product_normal}
\end{align}
\end{lemma}

\begin{proof}[Proof sketch]
Product tangent/normal rules followed by the inverse-image, intersection, pullback, and sum rules yield \eqref{eq:product_tangent}--\eqref{eq:product_normal}.
\end{proof}

Thus each blocked direction decomposes into action-set and primitive normal components, identifying the responsible block.

For a given problem, the designer selects the required blocks, while the transport layer receives only the compiled interface
$(\bphi_{\mathfrak{N}},\cK_{\mathfrak{N}},J_{\Phi})$.

\subsection{Long-Term Constraints as Lifted Dynamics}

Temporal budgets such as average delay, energy, loss, reliability, or stability can be lifted with a virtual state:
\begin{align}
&\limsup_{T\to\infty}\frac1T\sum_{t=0}^{T-1}\E[c(x_t,a_t)]\le b,
\label{eq:time_average_constraint}\\
&y_{t+1}=[y_t+c(x_t,a_t)-b]^+.
\label{eq:virtual_queue}
\end{align}

Hereafter, $x_t$ denotes the augmented Markov state containing $y_t$.
With $V(y)=\frac{1}{2}\|y\|^2$, define
$
\bphi^{\mathrm{dyn}}(x_t,a_t)
=
\E\!\left[
V(y_{t+1})-V(y_t)
\mid x_t,a_t
\right]
-
\delta_{\mathrm{dyn}}(x_t).
$
The inclusion $\bphi^{\mathrm{dyn}}(x_t,a_t)\in\R_-$ provides a one-step interface for accumulated constraint pressure, but does not imply
\eqref{eq:time_average_constraint} without virtual-state stability.

\begin{proposition}[Virtual-state implication]
\label{prop:virtual_queue_implication}
If \eqref{eq:virtual_queue} holds, $y_0$ and $c(x_t,a_t)$ are integrable, and $\E[\|y_T\|_1]/T\to0$, then \eqref{eq:time_average_constraint} holds componentwise.
\end{proposition}

\begin{proof}[Proof sketch]
Since $[z]^+\ge z$, $\sum_{t<T}c(x_t,a_t)\le Tb+y_T-y_0$. Taking expectations, dividing by $T$, and using mean-rate stability proves \eqref{eq:time_average_constraint}.
\end{proof}

Thus, the virtual state Markovizes accumulated constraint pressure, while the original average guarantee follows only after mean-rate stability is established.

\section{Learning Algorithm}
\label{sec:algorithm}

NFG inserts transport between the proto-actor and environment; replay, targets, and critics use transported actions while the proto-policy remains arbitrary.

\begin{algorithm}[t]
\caption{NFG-RL with feasibility transport}
\label{alg:nfg}
\begin{algorithmic}[1]
\Require Compiler, proto-policy $\rho_{\theta}$, critic $Q_{\psi}$, targets, replay $\mathcal D$, transport $\Gamma_{\mathfrak N,x}$
\For{each environment step $t$}
    \State Observe $x_t$, compile the residual interface, sample $u_t\sim\rho_{\theta}(\cdot\mid x_t)$, and transport $a_t=\Gamma_{\mathfrak N,x_t}(u_t)$ using $g_t=\nabla_aQ_{\psi}(x_t,a)|_{a=\bar a_t}$.
    \State Execute $a_t$, observe $r_t=U(x_t,a_t)$ and $x_{t+1}$, and store $(x_t,a_t,r_t,x_{t+1})$ in $\mathcal{D}$.
    \State Sample minibatch $\mathcal{B}\subset\mathcal{D}$.
    \For{each $(x_i,a_i,r_i,x'_i)\in\mathcal{B}$}
        \State Compile at $x'_i$, sample $u'_i\sim\rho_{\bar\theta}(\cdot\mid x'_i)$, and compute $a'_i$ by \eqref{eq:transport} using $Q_{\bar\psi}$ and a target reference $\bar a'_i$.
        \State Set $\hat y_i=r_i+\gamma Q_{\bar\psi}(x'_i,a'_i)$.
    \EndFor
    \State Take a critic step on $\mathcal{J}_Q(\psi)=|\mathcal{B}|^{-1}\sum_{i\in\mathcal{B}}(Q_{\psi}(x_i,a_i)-\hat y_i)^2$.
    \State Ascend $|\mathcal{B}|^{-1}\sum_{i\in\mathcal{B}} Q_{\psi}(x_i,\Gamma_{\mathfrak{N},x_i} (u_{\theta}(x_i,\epsilon_i)))$ in $\theta$, with $\epsilon_i\sim p$, and update the target networks.
\EndFor
\end{algorithmic}
\end{algorithm}

The reference $\bar a_t$ avoids evaluating $Q_{\psi}$ or $J_{\Phi}$ at an infeasible proposal and may be clipped, decoded, previously feasible, or a solver iterate; set $\bar a_t=u_t$ only when both maps extend there. Replay and targets use transported actions. Because the implicit pushforward density is unavailable, NFG-RL uses the reparameterized gradient \eqref{eq:transport_chain_rule}, not $\log\pi_{\theta}^{\mathfrak N}$.

\subsection{Transport Implementations}

NFG supports four regimes through the same interface. \emph{Exact transport} differentiates a convex KKT solve and is feasible when the solver succeeds. \emph{Unrolled flow} applies $K$ correction steps, trading latency and memory for residual accuracy. \emph{Tangent transport} uses $P_T^M$ for cheap first-order feasibility near the feasible set or before repair. \emph{Hybrid discrete transport} trains on a relaxation and executes structured rounding plus deterministic repair. A KKT solve with $n$ actions and $m_{\mathrm{act}}$ active rows costs $O((n+m_{\mathrm{act}})^3)$ and reuses its factorization.

\section{Theoretical Properties}
\label{sec:theory}
By Lemmas~\ref{lem:transport_compatibility} and \ref{lem:as_execution_feasibility}, exact transport satisfies
$
\Pr\!\left(
a_t\in\cF_{\mathfrak{N}}(x_t)
\ \text{for all } t\ge0
\right)=1,
$
with the analogous statement for $\cF_{\mathfrak{N}}^{\varepsilon}$ under $\varepsilon$-feasible transport; the executed support is restricted before any reward is observed. The remaining results are local: active constraints contract exploration in normal directions, the critic tilt yields a first-order gain over plain projection, and the dynamic primitive recovers backpressure. They do not imply global actor optimality or remove function-approximation error.

\begin{proposition}[Active constraints reshape exploration]
\label{prop:covariance}
Consider a regular feasible point
$a^\star=\Gamma_{\mathfrak{N},x}(u^\star)$.
Assume LICQ and strict complementarity hold at $a^\star$, that the active constraint functions are affine near $a^\star$, so that the active set is locally stable and $\cF_{\mathfrak{N}}(x)$ locally coincides with the manifold $J_A(a-a^\star)=0$, and that $M_{\mathfrak{N}}(x,u)$ and $g_{\psi}(x,u)$ are locally constant near $u^\star$.
If the proto-action has local covariance $\Sigma_{\theta}$ in this
neighborhood, then, to first order in the proto-action perturbation,
$
\Cov[a\mid x]
\approx
P_T^M\Sigma_{\theta}(P_T^M)^\top,
$
where $P_T^M$ is defined in \eqref{eq:projection}.
\end{proposition}

\begin{proof}[Proof sketch]
Completing the square in \eqref{eq:transport_objective} makes $\Gamma_{\mathfrak{N},x}$ the $M$-metric projection of $\tilde u_{\psi}=u+\eta M^{-1}g_{\psi}$. Local constancy gives $D_u\tilde u_{\psi}=I$, while the stable active manifold gives $D_u\Gamma_{\mathfrak{N},x}=P_T^M$.
Hence $a-a^\star\approx P_T^M(u-u^\star)$, and first-order covariance
propagation yields the stated relation.
\end{proof}

Proposition~\ref{prop:covariance} shows that NFG reshapes exploration without an extra covariance regularizer: normal directions to active constraints are contracted, while tangent directions remain available for policy improvement.
\begin{proposition}[First-order improvement over plain projection]
\label{prop:improvement}
Fix $x$ and $u$. Let $a_0$ be the plain $M$-metric projection obtained from \eqref{eq:transport} with $g_{\psi}=0$, and let $a_{\eta}$ use the fixed critic direction $g=\nabla_aQ_{\psi}(x,a_0)$ with step $\eta>0$. Suppose the assumptions of Proposition~\ref{prop:covariance} hold at $a_0$ and $Q_{\psi}(x,\cdot)$ is continuously differentiable nearby. Then
$
a_{\eta}
=a_0+\eta P_T^M M^{-1}g+o(\eta),
\left.\frac{d}{d\eta}Q_{\psi}(x,a_{\eta})\right|_{\eta=0^+}
 =\bigl\|(I-\Pi_M)M^{-1/2}g\bigr\|^2\ge0,
$
where
$
\Pi_M
=
M^{-1/2}J_A^{\top}
\bigl(J_AM^{-1}J_A^{\top}\bigr)^{\dagger}
J_AM^{-1/2}.
$
The gain is strict unless $g\in\operatorname{range}(J_A^{\top})$, i.e., unless $a_0$ is first-order stationary on the active manifold.
\end{proposition}

\begin{proof}[Proof sketch]
Completing the square makes $a_{\eta}$ the $M$-projection of $u+\eta M^{-1}g$. Local differentiation gives the expansion of $a_{\eta}$, and Taylor expansion yields
$
Q_{\psi}(x,a_{\eta})-Q_{\psi}(x,a_0)
=
\eta g^{\top}P_T^MM^{-1}g+o(\eta).
$
Because
$
P_T^M M^{-1}
=
M^{-1/2}(I-\Pi_M)M^{-1/2},
$
the first-order coefficient is the stated squared norm and vanishes exactly when $g\in\operatorname{range}(J_A^{\top})$.
\end{proof}

Proposition~\ref{prop:improvement} identifies the critic tilt as one step of $M$-metric projected ascent along the feasible manifold. Plain projection with $g_{\psi}=0$ forgoes this first-order gain, consistent with the projection-only ablation in Section~\ref{sec:experiments}.
\begin{proposition}[Backpressure as a dynamic feasibility gradient]
\label{prop:backpressure}
Consider a multi-hop queueing network away from the reflection boundary,
with
$
q_{i,k}^{+}
=
q_{i,k}
+
A_{i,k}
+
\sum_j f_{ji,k}
-
\sum_j f_{ij,k}.
$
Let $V(q)=\frac{1}{2}\sum_{i,k}q_{i,k}^2$.
Assume that exogenous terms are action-independent and that
differentiation may be interchanged with conditional expectation.
Then
$
\frac{\partial}{\partial f_{ij,k}}
\E\!\left[V(q^{+})-V(q)\mid q,a\right]
=
\E\!\left[q_{j,k}^{+}-q_{i,k}^{+}\mid q,a\right].
$
When the expected one-slot queue increments are small relative to the
current backlogs, the right-hand side is approximately
$q_{j,k}-q_{i,k}$.
Hence the negative dynamic-residual gradient recovers the differential
backlog direction $q_{i,k}-q_{j,k}$ used by backpressure
scheduling~\cite{tassiulas1990stability,neely2010stochastic}.
\end{proposition}

\begin{proof}[Proof sketch]
Away from reflection,
$\partial q^{+}/\partial f_{ij,k}=-e_{i,k}+e_{j,k}$.
Therefore, samplewise differentiation gives
$
\frac{\partial}{\partial f_{ij,k}}
\bigl(V(q^{+})-V(q)\bigr)
=
(q^{+})^\top(-e_{i,k}+e_{j,k})
=
q_{j,k}^{+}-q_{i,k}^{+}.
$
Taking conditional expectations yields the stated identity.
Writing $q^{+}=q+\Delta q$ gives the current-backlog approximation when
$\E[\Delta q\mid q,a]$ is small.
\end{proof}
Thus, backpressure emerges as the gradient of a lifted dynamic residual. Together with the support guarantee and Propositions~\ref{prop:covariance} and~\ref{prop:improvement}, this shows that network constraints shape the executed support, the exploration geometry, and the actor gradient of the policy.

\section{Case Study: Multi-Tenant Wireless Edge Computing}
\label{sec:application}

We now instantiate NFG in a multi-tenant wireless edge computing system, compiling routing, wireless scheduling, power, compute, placement, admission, and offloading into one feasibility interface.
After defining the state and joint action, we instantiate the residual blocks and show that their product set and stacked Jacobian require no change to the proto-actor or transport interface.

\subsection{System Model}

The network is $G=(\mathcal{V},\mathcal{E})$, with users, base stations, edge servers, aggregation switches, and cloud gateways as nodes, and $\mathcal{M}$ as the tenant or commodity set. Tenant $k\in\mathcal{M}$ follows the chain $\mathcal{S}_k=(\sigma_{k,s})_{s=1}^{L_k}$, whose stages cover communication and edge/cloud functions such as preprocessing, inference, and delivery.

At slot $t$, let $a_t=(f_t,\chi_t,p_t,z_t,\kappa_t,m_t)$, where the components denote routing rates, wireless activation, transmit power, placement, compute allocation, and admission/offloading gates, respectively.
We reserve $\ell$ for primitive types and $s$ for service stages, and suppress time indices below, so $m_k$ denotes the admission gate of tenant $k$.
The action set $\cA$ enforces domain and resource bounds together with continuous box/simplex relaxations; execution applies structured rounding and repair to integral components.

The Markov state $x_t$ contains observable topology, channels, queues, server load, tenant demand, service-chain and failure states, and any virtual state for long-term constraints.
Together with arrivals, channel and failure evolution, queue dynamics, and service execution, $(x_t,a_t)$ instantiates Definition~\ref{def:nds}.

A representative utility is $U(x_t,a_t)=\sum_{k\in\mathcal{M}}[\mathsf{u}_k(R_{k,t})-\beta^D_kD_{k,t}-\beta^S_kV^{\mathrm{SLA}}_{k,t}]-\beta^E\mathsf{E}_t$, where $R_{k,t}$, $D_{k,t}$, $\mathsf{E}_t$, and $V^{\mathrm{SLA}}_{k,t}$ denote throughput, delay, energy, and SLA violation.
Utility weights rank admissible actions, whereas $\bphi_{\mathfrak{N}}(x,a)\in\cK_{\mathfrak{N}}$ defines executability and hence the feasible policy class.

\subsection{Compiling the Edge Constraint Map}

Let $\mathcal{L}_{\mathrm{fab}}\subseteq\mathcal{L}$ index the fourteen constraint blocks defined below.
Their constraint map and admissible set are $\bphi_{\mathfrak{N}}(x,a)=(\bphi^{\ell}(x,a))_{\ell\in\mathcal{L}_{\mathrm{fab}}}$ and $\cK_{\mathfrak{N}}=\prod_{\ell\in\mathcal{L}_{\mathrm{fab}}}\cK^{\ell}$, with $J_{\Phi}(x,a)=\operatorname{col}_{\ell\in\mathcal{L}_{\mathrm{fab}}}D_a\bphi^{\ell}(x,a)$.
Thus product membership is equivalent to simultaneous block admissibility, while transport uses only the compiled triple $(\bphi_{\mathfrak{N}},\cK_{\mathfrak{N}},J_{\Phi})$.

\paragraph{Flow conservation and service-chain continuity.}
For $k\in\mathcal{M}$, use $\bphi^{\mathrm{bal}}_k(x,a)=B_kf_k-b_k(x,m)\in\{0\}$ and $\bphi^{\mathrm{svc}}(x,a)=A_{\mathrm{svc}}z+B_{\mathrm{svc}}f-b_{\mathrm{svc}}(x,m)\in\{0\}$, where $B_k$ is the possibly service-expanded incidence operator and the $b$-terms encode state-dependent admitted demand and stage requirements.
These blocks couple routing, placement, and admission, and their Jacobian rows suppress directions violating topology or service order.

\paragraph{Wireless capacity, power, and interference.}
For $e\in\mathcal{E}$, let $\mathsf{SINR}_e(x,\chi,p)=p_eh_e/(\sum_{e'\in\mathcal{I}(e)}p_{e'}h_{e'\to e}+N_0)$ and $C_e(x,\chi,p)=W_e\chi_e\log(1+\mathsf{SINR}_e)$, where $W_e$, $h_e$, $h_{e'\to e}$, $\mathcal{I}(e)$, and $N_0$ denote bandwidth, direct and cross gains, interferers, and noise.
Capacity and activation--power coupling use $\bphi^{\mathrm{cap}}_e(x,a)=\sum_{k\in\mathcal{M}}\sum_{s=1}^{L_k}f_{e,k,s}-C_e(x,\chi,p)\in\R_-$ and $\bphi^{\mathrm{pow}}_e(x,a)=p_e-P_e^{\max}\chi_e\in\R_-$, where $p_e$ is actual transmit power, so inactive links have $p_e=0$.
Protocol interference uses $\bphi^{\mathrm{phy}}_{\mathrm{conf}}(x,a)=M_{\mathrm{int}}\chi-\one\in\R_-^{d_{\mathrm{conf}}}$, where $M_{\mathrm{int}}$ has $d_{\mathrm{conf}}$ conflict-set rows; under an SINR model, use $\bphi^{\mathrm{phy}}_e(x,a)=\chi_e\gamma_e^{\min}-\mathsf{SINR}_e(x,\chi,p)\in\R_-$.
These blocks expose routing, activation, interference, and power sensitivities to $J_{\Phi}$.

\paragraph{Queue stability.}
Let $q_{t+1}=F_q(q_t,a_t,\omega_t)$ and $V(q)=\frac{1}{2}\|q\|^2$.
For a prescribed state-dependent drift upper bound $\delta_q(x_t)$, typically negative outside a bounded queue set, use $\bphi^{\mathrm{dyn}}(x_t,a_t)=\E[V(q_{t+1})-V(q_t)\mid x_t,a_t]-\delta_q(x_t)\in\R_-$.
This lifted block injects backlog pressure into transport; as in Section~\ref{sec:primitives}, long-term stability additionally requires queue or virtual-state stability.

\paragraph{Edge compute and memory.}
For server $v\in\mathcal{V}$, use $\bphi^{\mathrm{cpu}}_v(x,a)=\sum_{k\in\mathcal{M}}\sum_{s=1}^{L_k}\kappa^{\mathrm{cpu}}_{v,k,s}-C_v^{\mathrm{cpu}}\in\R_-$, $\bphi^{\mathrm{cpu\text{-}place}}_{v,k,s}(x,a)=\kappa^{\mathrm{cpu}}_{v,k,s}-C_v^{\mathrm{cpu}}z_{v,k,s}\in\R_-$, and $\bphi^{\mathrm{mem}}_v(x,a)=\sum_{k\in\mathcal{M}}\sum_{s=1}^{L_k}d^{\mathrm{mem}}_{k,s}z_{v,k,s}-C_v^{\mathrm{mem}}\in\R_-$.
The second block forbids compute allocation to an unplaced stage; together, the three blocks couple placement with server resources.

\paragraph{Slicing and fairness.}
Minimum rate and normalized share use $\bphi^{\mathrm{slice}}_k(x,a)=m_kR_k^{\min}-R_k(x,a)\in\R_-$ and $\bphi^{\mathrm{share}}_k(x,a)=\zeta_km_k-R_k(x,a)/(\tau_R+\sum_{j\in\mathcal{M}}R_j(x,a))\in\R_-$, where $\zeta_k\ge0$ is the configured minimum share and $\tau_R>0$ prevents division by zero.
At execution $m_k\in\{0,1\}$, while its relaxation continuously gates admission.
Summing the share constraints gives the necessary condition $\sum_{k\in\mathcal{M}}\zeta_km_k<1$.

\paragraph{Tail latency and reliability.}
Use $\bphi^{\mathrm{risk}}_k(x,a)=m_k[\CVaR_{\alpha}(D_k(x,a,\omega))-D_k^{\max}]\in\R_-$ and $\bphi^{\mathrm{rel}}_k(x,a)=m_k[P_k^{\min}-P_k^{\mathrm{succ}}(x,a)]\in\R_-$.
The gate deactivates both constraints for rejected tenants, while admitted tenants contribute the corresponding gradients or subgradients to the same transport interface~\cite{rockafellar2000optimization}.

\paragraph{Discrete decisions.}
Placement, path, channel, and admission variables use the closed training relaxation $a_{\mathrm{disc}}\in\operatorname{conv}(\mathcal{I}_{\mathfrak{N}}(x))$, where $\mathcal{I}_{\mathfrak{N}}(x)$ is the integral feasible set; when available, its exact half-space representation is $\bphi^{\mathrm{disc}}(x,a)=G_{\mathrm{disc}}(x)a_{\mathrm{disc}}-h_{\mathrm{disc}}(x)\in\R_-^{d_{\mathrm{disc}}}$, with $d_{\mathrm{disc}}$ inequalities.
At execution, structured rounding and repair must return $a_{\mathrm{disc}}^{\mathrm{exec}}=\mathcal{R}_x(\operatorname{round}(a_{\mathrm{disc}}))\in\mathcal{I}_{\mathfrak{N}}(x)$; otherwise only relaxed or $\varepsilon$-compatibility may be claimed.

\begin{proposition}[Edge-fabric constraint map consistency]
\label{prop:edge_law_map_consistency}
Assume $\cA$ is closed and the stacked feasible set is nonempty.
For every $\ell\in\mathcal{L}_{\mathrm{fab}}$, assume $\bphi^{\ell}$ is jointly measurable, continuous in $a$, and differentiable, or locally Lipschitz with a selected generalized-Jacobian element, where $J_{\Phi}$ is used, and that $\cK^{\ell}$ is nonempty and closed.
For discrete components, assume execution returns an action in $\mathcal{I}_{\mathfrak{N}}(x)$, or restrict the claim to the closed relaxation.
Then the compiled edge-fabric model satisfies Assumption~\ref{ass:basic_regularity}.
If $\Gamma_{\mathfrak{N},x}(u)\in\cF_{\mathfrak{N}}(x)$ $\rho_\theta(\cdot\mid x)$-a.s. for every $x$, then $\pi_\theta^{\mathfrak{N}}\in\Pi_{\mathfrak{N}}$; $\varepsilon$-feasible transport analogously gives $\pi_\theta^{\mathfrak{N}}\in\Pi_{\mathfrak{N}}^{\varepsilon}$.
\end{proposition}

\begin{proof}[Proof sketch]
Finite stacking preserves joint measurability and continuity, while finite products preserve nonemptiness and closedness.
Lemma~\ref{lem:primitive_composition} identifies the stacked feasible set with simultaneous block membership, establishing Assumption~\ref{ass:basic_regularity}.
Lemma~\ref{lem:transport_compatibility} then gives the exact or $\varepsilon$-compatible pushforward claim; exact integrality additionally requires the stated repair membership.
\end{proof}

The proposition shows that these requirements are blocks of one compiled interface rather than separate penalties or controllers.
A deployment may select any subset, but the proto-actor and transport layer still receive only $(\bphi_{\mathfrak{N}},\cK_{\mathfrak{N}},J_{\Phi})$.

\section{Experiments}
\label{sec:experiments}

We evaluate NFG-RL in one controlled and two public-trace-conditioned wireless-edge surrogates, SMEC-5G and 5G-C3. All environments retain the joint action $a_t=(f_t,\chi_t,p_t,z_t,\kappa_t,m_t)$ and the residual families in Section~\ref{sec:application}, and evaluate utility--feasibility tradeoffs, component contributions, and robustness to load and system shifts~\cite{zhang2026enabling,raca2020beyond,tocze2022edge,mehran2022matching}.

\subsection{Common Evaluation Protocol}
\label{subsec:common_protocol}

\paragraph{Compared methods.}
We compare NFG-RL with Greedy-EDF; queue-aware BP-DPP
~\cite{tassiulas1990stability,neely2010stochastic};
PPO- and SAC-Penalty~\cite{schulman2017proximal,haarnoja2018soft};
PPO-Lagrangian and CPO
~\cite{altman2021constrained,tessler2018reward,achiam2017constrained};
GNN-SAC~\cite{rusek2020routenet,almasan2022deep,geyer2019deeptma};
and OptLayer-SAC~\cite{amos2017optnet,agrawal2019differentiable}.
The reduced-instance ``Myopic-Search'' baseline is a myopic enumerative reference
that searches binary admission and placement choices, heuristically
allocates wireless and compute shares, and maximizes one-step model
utility; it is neither an upper bound nor a dynamic optimum. All methods
use the same 27-dimensional state and 18-dimensional action schema, and
learned methods share training budgets and network widths. GNN-SAC uses
a three-node, one-step tenant-interaction encoder and is therefore a
structured-encoder, rather than a physical-topology, baseline.

\paragraph{Raw and executed actions.}
Each method proposes a raw action $a_t$. Only proposals that would make
the simulator undefined are minimally repaired as
$\widetilde a_t=\mathcal{R}_{x_t}(a_t)$ through bound clipping, conflict
removal, excess-traffic dropping, and integral-placement repair.
Raw-action violations and $\|\widetilde a_t-a_t\|_2$ are evaluated before
execution, whereas utility, throughput, and delay use
$\widetilde a_t$, preventing simulator-side repair from hiding
infeasibility. The deployed NFG transport uses finite unrolling, with deterministic tie breaking for integral decisions.

\paragraph{Metrics.}
Raw-action feasibility is measured by
$\Delta_{\mathrm{res}}(x_t,a_t)
=\dist_{\mathcal Z}(\bphi_{\mathfrak{N}}(x_t,a_t),
\cK_{\mathfrak{N}})$ and its blockwise counterpart
$\Delta_{\ell}(x_t,a_t)
=\dist_{\mathcal Z^{\ell}}(\bphi^{\ell}(x_t,a_t),\cK^{\ell})$.
We report feasible utility, SLA satisfaction, tail delay, raw-action
violation, residual and repair distances, starvation, and decision
latency. Feasible utility subtracts normalized emergency-drop and repair
penalties from executed-action utility.
Because deployment uses finitely many unrolled transport steps, NFG-RL is evaluated as $\varepsilon$-compatible under Lemma~\ref{lem:transport_compatibility}; its nonzero raw-action violation is the remaining solver residual, not a contradiction of the exact-transport support guarantee.

\paragraph{Training and statistical reporting.}
Learned methods use eleven random seeds ($0$--$10$), batch size 128, hidden width 128,
replay capacity $10^5$, and a 96-slot horizon. Main and ablation training
use 1000 and 600 steps, respectively, while nominal, common-stress, and
ablation-stress evaluation uses 10, 3, and 2 episodes per seed. Tables
report mean $\pm$ standard deviation across seeds. Training uses six
unrolled transport steps with periodic active-set warm starts. The main
NFG-RL results use the model-assisted active-set evaluator, whereas the
16-step unrolled evaluator is used only for runtime diagnostics; hence,
the main comparison is model-assisted rather than purely model-free.

\subsection{Part I: Controlled Application Simulation}
\label{subsec:controlled_simulation}

\subsubsection{Controlled Three-Tenant Application Surrogate}

The controlled environment is a three-tenant reduction of Section~\ref{sec:application}: it retains the six action groups and fourteen residual blocks but omits explicit multi-hop and per-stage service-chain topologies. Its 27-dimensional state includes demand, channel proxies, queues, deadlines and class weights, CPU and memory availability, background load,
mobility and failures, time features, load, and environment identity.
A 4096-point, eight-dimensional noisy sinusoidal sequence drives demand
(first three coordinates), channel quality (next three), background load,
and mobility (last two). All methods share the same differentiable wireless/interference,
edge/cloud compute, queue, delay, SLA, energy, starvation, utility, and
residual dynamics.

Experiments use loads $\rho\in\{0.5,0.7,0.9,1.1,1.3\}$ and a raw-action violation tolerance
of $10^{-4}$; the remaining training and evaluation settings follow the common protocol above.

\subsubsection{Feasibility and Load Stress}

All methods are trained at $\rho=0.9$ and evaluated over $\rho\in\{0.5,0.7,0.9,1.1,1.3\}$ without retuning. OptLayer-SAC isolates generic Euclidean projection, while the penalty/Lagrangian baselines isolate scalar constraint handling.

\begin{table}[t]
\centering
\caption{Nominal controlled results at $\rho=0.9$
(mean $\pm$ standard deviation).}
\label{tab:simulation_main}
\scriptsize
\setlength{\tabcolsep}{2.2pt}
\begin{tabular}{@{}lccc@{}}
\toprule
Method
& Utility $\uparrow$
& P95 delay $\downarrow$
& Violation $\downarrow$ \\
\midrule
Greedy-EDF
& $2.80\!\pm\!0.56$
& $442.17\!\pm\!159.59$
& $0.236\!\pm\!0.012$ \\
BP-DPP
& $2.86\!\pm\!0.55$
& $577.13\!\pm\!154.44$
& $0.318\!\pm\!0.015$ \\
MILP-Oracle
& $9.09\!\pm\!0.16$
& $208.50\!\pm\!6.48$
& $0.262\!\pm\!0.023$ \\
PPO-Penalty
& $0.46\!\pm\!0.64$
& $253.91\!\pm\!127.06$
& $0.623\!\pm\!0.032$ \\
SAC-Penalty
& $5.33\!\pm\!4.08$
& $231.30\!\pm\!149.00$
& $0.663\!\pm\!0.023$ \\
PPO-Lagr.
& $1.21\!\pm\!1.43$
& $251.75\!\pm\!135.41$
& $0.620\!\pm\!0.034$ \\
CPO
& $0.02\!\pm\!0.57$
& $249.20\!\pm\!137.25$
& $0.626\!\pm\!0.028$ \\
GNN-SAC
& $6.85\!\pm\!4.07$
& $134.14\!\pm\!26.58$
& $0.524\!\pm\!0.048$ \\
OptLayer-SAC
& $1.17\!\pm\!1.26$
& $810.24\!\pm\!968.01$
& $0.625\!\pm\!0.075$ \\
NFG-RL
& $\mathbf{12.83\!\pm\!0.33}$
& $\mathbf{128.25\!\pm\!31.35}$
& $\mathbf{0.189\!\pm\!0.005}$ \\
\bottomrule
\end{tabular}
\end{table}

\begin{table}[t]
\centering
\caption{Ablations at nominal load and under the six-shift suite
(mean $\pm$ standard deviation).}
\label{tab:ablation}
\scriptsize
\setlength{\tabcolsep}{2.4pt}
\begin{tabular}{@{}lccc@{}}
\toprule
Variant
& Nominal $U$ $\uparrow$
& Shift $U$ $\uparrow$
& Repair $\downarrow$ \\
\midrule
Full NFG-RL
& $\mathbf{12.83\!\pm\!0.33}$
& $9.68\!\pm\!1.56$
& $\mathbf{0.000\!\pm\!0.000}$ \\
w/o flow/service
& $11.53\!\pm\!0.26$
& $8.38\!\pm\!1.71$
& $0.147\!\pm\!0.035$ \\
w/o link/wireless
& $8.16\!\pm\!1.26$
& $6.94\!\pm\!1.75$
& $0.000\!\pm\!0.000$ \\
w/o queue dyn.
& $12.19\!\pm\!0.11$
& $9.24\!\pm\!1.78$
& $0.000\!\pm\!0.000$ \\
w/o compute/mem.
& $10.73\!\pm\!0.46$
& $7.56\!\pm\!1.60$
& $0.161\!\pm\!0.036$ \\
w/o slicing/fair.
& $12.47\!\pm\!0.11$
& $\mathbf{9.77\!\pm\!1.67}$
& $0.000\!\pm\!0.000$ \\
w/o risk/reliab.
& $12.47\!\pm\!0.13$
& $9.71\!\pm\!1.75$
& $0.000\!\pm\!0.000$ \\
w/o discrete feas.
& $2.22\!\pm\!1.59$
& $-5.82\!\pm\!6.61$
& $0.290\!\pm\!0.069$ \\
Isotropic metric
& $12.45\!\pm\!0.23$
& $8.97\!\pm\!1.07$
& $0.000\!\pm\!0.000$ \\
Projection only
& $-10.30\!\pm\!9.98$
& $-18.79\!\pm\!14.64$
& $0.416\!\pm\!0.126$ \\
Stop-gradient
& $12.17\!\pm\!0.13$
& $9.32\!\pm\!1.70$
& $0.000\!\pm\!0.000$ \\
\bottomrule
\end{tabular}
\end{table}

\begin{table*}[h]
\centering
\caption{Results in the two public-trace-conditioned surrogate environments
(mean $\pm$ standard deviation over seeds). The best and second-best results
in each environment are shown in bold and underlined, respectively.}
\label{tab:real_trace_results}
\scriptsize
\setlength{\tabcolsep}{1.8pt}
\renewcommand{\arraystretch}{1.08}
\resizebox{\textwidth}{!}{%
\begin{tabular}{@{}lcccccccccccc@{}}
\toprule
&
\multicolumn{6}{c}{SMEC-5G}
&
\multicolumn{6}{c}{5G-C3}
\\
\cmidrule(lr){2-7}
\cmidrule(lr){8-13}
Method
& Utility $\uparrow$
& SLA $\uparrow$
& P99 $\downarrow$
& Viol. $\downarrow$
& Repair $\downarrow$
& Starv. $\downarrow$
& Utility $\uparrow$
& SLA $\uparrow$
& P99 $\downarrow$
& Viol. $\downarrow$
& Repair $\downarrow$
& Starv. $\downarrow$
\\
\midrule

Greedy-EDF
& $4.55 \pm 0.89$
& $0.15 \pm 0.02$
& $459.54 \pm 468.84$
& $\underline{0.205 \pm 0.010}$
& $\underline{0.031 \pm 0.005}$
& $0.410 \pm 0.013$
& $3.16 \pm 0.20$
& $0.09 \pm 0.03$
& $\underline{117.30 \pm 5.05}$
& $0.214 \pm 0.009$
& $\underline{0.042 \pm 0.003}$
& $0.586 \pm 0.010$
\\

BP-DPP
& $5.82 \pm 1.49$
& $0.28 \pm 0.04$
& $571.34 \pm 631.20$
& $0.234 \pm 0.026$
& $0.069 \pm 0.024$
& $0.364 \pm 0.023$
& $5.23 \pm 0.82$
& $0.36 \pm 0.09$
& $117.77 \pm 18.74$
& $\underline{0.198 \pm 0.021}$
& $0.056 \pm 0.025$
& $0.370 \pm 0.036$
\\

MILP-Oracle
& $7.03 \pm 0.15$
& $0.36 \pm 0.03$
& $\underline{127.55 \pm 11.21}$
& $0.378 \pm 0.010$
& $0.350 \pm 0.017$
& $0.327 \pm 0.006$
& $\underline{6.90 \pm 0.18}$
& $\underline{0.39 \pm 0.03}$
& $327.02 \pm 53.51$
& $0.285 \pm 0.023$
& $0.059 \pm 0.019$
& $0.170 \pm 0.056$
\\

PPO-Penalty
& $0.14 \pm 0.52$
& $0.13 \pm 0.06$
& $212.74 \pm 139.77$
& $0.619 \pm 0.011$
& $0.461 \pm 0.005$
& $0.653 \pm 0.013$
& $0.35 \pm 0.53$
& $0.18 \pm 0.08$
& $219.81 \pm 103.78$
& $0.604 \pm 0.028$
& $0.443 \pm 0.014$
& $0.552 \pm 0.034$
\\

SAC-Penalty
& $7.39 \pm 1.25$
& $\underline{0.40 \pm 0.05}$
& $215.16 \pm 190.43$
& $0.604 \pm 0.039$
& $0.320 \pm 0.064$
& $\underline{0.064 \pm 0.056}$
& $5.04 \pm 1.88$
& $0.32 \pm 0.07$
& $608.12 \pm 526.08$
& $0.647 \pm 0.017$
& $0.422 \pm 0.048$
& $0.120 \pm 0.090$
\\

PPO-Lagrangian
& $0.10 \pm 0.60$
& $0.16 \pm 0.05$
& $261.45 \pm 208.25$
& $0.620 \pm 0.012$
& $0.458 \pm 0.011$
& $0.631 \pm 0.023$
& $0.53 \pm 0.83$
& $0.18 \pm 0.08$
& $286.84 \pm 159.79$
& $0.603 \pm 0.033$
& $0.439 \pm 0.018$
& $0.521 \pm 0.063$
\\

CPO
& $-0.09 \pm 0.49$
& $0.15 \pm 0.06$
& $264.80 \pm 212.41$
& $0.620 \pm 0.012$
& $0.458 \pm 0.010$
& $0.650 \pm 0.024$
& $-0.09 \pm 0.94$
& $0.16 \pm 0.09$
& $159.12 \pm 59.04$
& $0.590 \pm 0.043$
& $0.449 \pm 0.016$
& $0.604 \pm 0.076$
\\

GNN-SAC
& $6.97 \pm 0.09$
& $0.33 \pm 0.03$
& $148.38 \pm 21.28$
& $0.547 \pm 0.043$
& $0.356 \pm 0.008$
& $0.077 \pm 0.013$
& $5.65 \pm 1.22$
& $0.19 \pm 0.09$
& $247.16 \pm 90.58$
& $0.549 \pm 0.019$
& $0.405 \pm 0.042$
& $\underline{0.107 \pm 0.060}$
\\

OptLayer-SAC
& $\underline{7.55 \pm 1.29}$
& $0.25 \pm 0.08$
& $216.13 \pm 158.17$
& $0.415 \pm 0.067$
& $0.269 \pm 0.031$
& $0.081 \pm 0.048$
& $3.43 \pm 1.69$
& $0.20 \pm 0.05$
& $826.42 \pm 731.49$
& $0.519 \pm 0.027$
& $0.340 \pm 0.022$
& $0.283 \pm 0.132$
\\

NFG-RL
& $\mathbf{10.38 \pm 0.11}$
& $\mathbf{0.49 \pm 0.05}$
& $\mathbf{92.96 \pm 12.40}$
& $\mathbf{0.163 \pm 0.012}$
& $\mathbf{0.000 \pm 0.000}$
& $\mathbf{0.006 \pm 0.002}$
& $\mathbf{9.76 \pm 0.15}$
& $\mathbf{0.70 \pm 0.07}$
& $\mathbf{80.04 \pm 0.96}$
& $\mathbf{0.146 \pm 0.003}$
& $\mathbf{0.000 \pm 0.000}$
& $\mathbf{0.041 \pm 0.012}$
\\

\bottomrule
\end{tabular}%
}
\end{table*}

\begin{figure}[!t]
\centering
\begin{minipage}[t]{0.495\columnwidth}
    \centering
    \includegraphics[width=\linewidth]
    {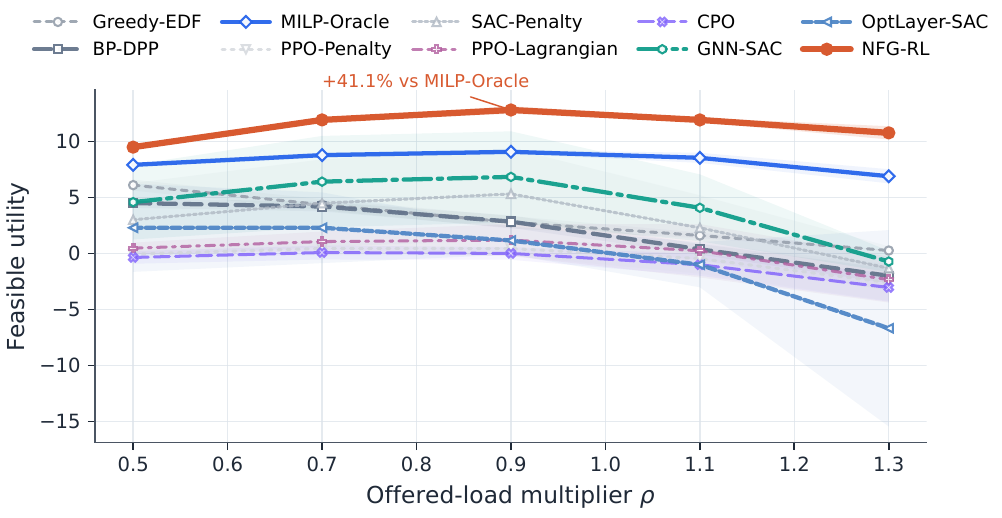}
    \vspace{-1.5mm}

    {\scriptsize (a) Controlled load stress.}
\end{minipage}\hfill
\begin{minipage}[t]{0.495\columnwidth}
    \centering
    \includegraphics[width=\linewidth]
    {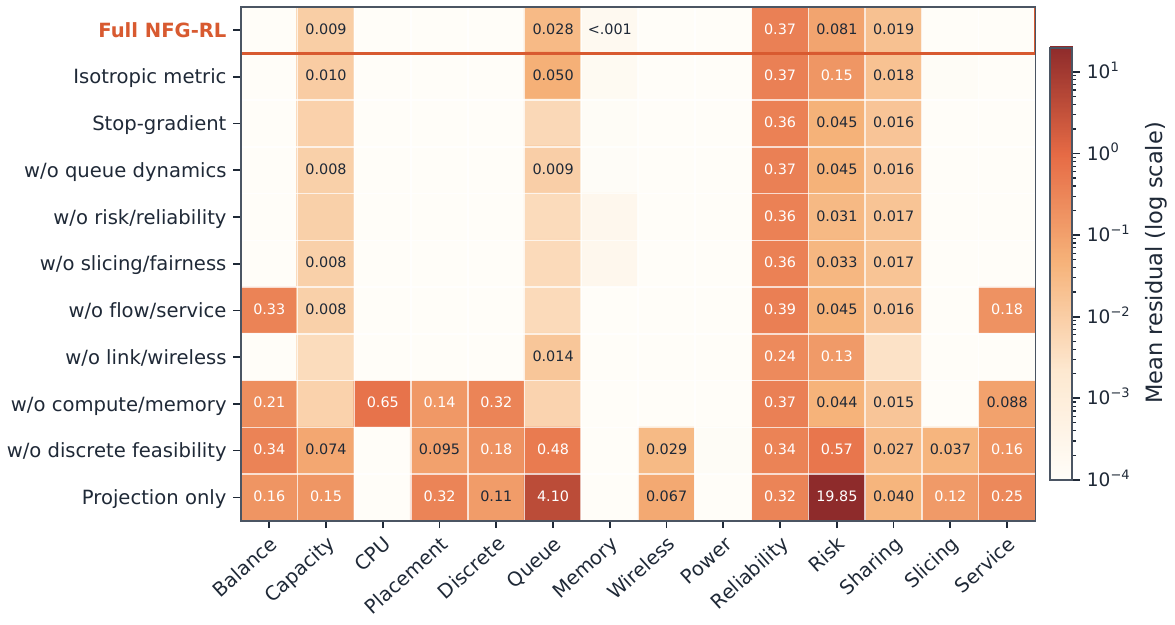}
    \vspace{-1.5mm}

    {\scriptsize (b) Residual decomposition.}
\end{minipage}
\vspace{-1mm}
\caption{Controlled evaluation: (a) feasible utility across load and
(b) blockwise residual magnitudes for NFG-RL and its ablations on a shared
logarithmic scale.}
\label{fig:controlled_results}
\vspace{-2mm}
\end{figure}

\noindent\textbf{Controlled results.}
At $\rho=0.9$, NFG-RL achieves utility $12.83$ versus $9.09$ for the strongest non-NFG method, MILP-Oracle ($41.1\%$ higher), P95 delay $128.25$ versus $134.14$ for GNN-SAC ($4.4\%$ lower), and violation $0.189$ versus $0.236$ for Greedy-EDF ($19.9\%$ lower). Its SLA satisfaction, residual distance, repair distance, and decision latency are $0.27\!\pm\!0.04$, $0.035\!\pm\!0.003$, $0$, and $1.20\!\pm\!0.06$\,ms, respectively. The myopic oracle requires
$14.99\!\pm\!0.08$\,ms, whereas the remaining non-oracle baselines use $0.05$--$0.25$\,ms. At $\rho=1.3$, NFG-RL retains the highest observed utility and records P95 delay $196.38$ and P95 residual distance $0.266$
[Fig.~\ref{fig:controlled_results}(a)].

\subsubsection{Constraint and Mechanism Ablations}
\label{subsubsec:constraint_ablation}

Each constraint ablation removes one residual family from transport while retaining it in evaluation: flow/service, link/wireless, queue dynamics, compute/memory, slicing/fairness, risk/reliability, or discrete feasibility. Mechanism ablations set $\varsigma=0$ in~\eqref{eq:metric} (isotropic), $g_{\psi}=0$ in~\eqref{eq:transport_objective} (projection only), or block $D_u\Gamma$ in~\eqref{eq:transport_chain_rule} (stop-gradient).

\paragraph{Ablation results.}
Full NFG-RL has the highest nominal utility, $2.9\%$ above the closest ablation. Removing discrete feasibility reduces nominal/shift-average utility from $12.83/9.68$ to $2.22/-5.82$ and raises repair from $0$ to $0.290$, while projection only is unstable at $-10.30/-18.79$.
Removing flow/service or compute/memory raises shift-average violation from $0.194$ to $0.344$ and $0.551$, respectively, and repair to $0.147$ and $0.161$, consistent with the localized residuals in Fig.~\ref{fig:controlled_results}(b). Removing slicing/fairness or
risk/reliability instead yields shift-average utilities $9.77$ and
$9.71$, slightly above $9.68$ for the full model; thus, no single
residual family uniformly maximizes unweighted utility under every
perturbation.

\subsection{Part II: Public-Trace-Conditioned Surrogate Environments}
\label{subsec:real_trace}

\subsubsection{Data Processing and Environment Construction}

SMEC-5G uses numeric records from the SMEC artifact, while 5G-C3 combines the UCC 5G channel/context trace with C3 edge-infrastructure records~\cite{zhang2026enabling,raca2020beyond,tocze2022edge,
mehran2022matching}. These are trace-conditioned surrogates rather than
direct deployment replays. For each environment, we retain at most 64
numeric columns, concatenate them, select the eight highest-variance
dimensions, scale them by their 2nd and 98th percentiles, clip them to
$[0,1]$, and resample them to 4096 points. The resulting coordinates
condition demand (three), channel quality (three), background load, and
mobility, without assuming a one-to-one correspondence with named
measurement fields.

Both environments share the state, action, transition, residual, and
safety models. SMEC-5G applies tenant-dependent demand scaling,
mobility-dependent channel attenuation, reduced CPU availability,
tighter deadlines, and environment-specific class weights; 5G-C3 uses
tenant-dependent demand scaling, trace-conditioned CPU availability,
reduced memory, and environment-specific deadlines. Seed- and
episode-indexed windows are deterministic, with no separate
chronological split.

\subsubsection{Public-Trace-Conditioned Comparison}

All ten methods are evaluated at $\rho=0.9$ for 10 episodes per seed. The reported metrics are generated by the common surrogate dynamics conditioned on the processed traces and are not all direct measurements from the source datasets.

\begin{figure}[!t]
\centering
\begin{minipage}[t]{0.495\columnwidth}
    \centering
    \includegraphics[width=\linewidth]
    {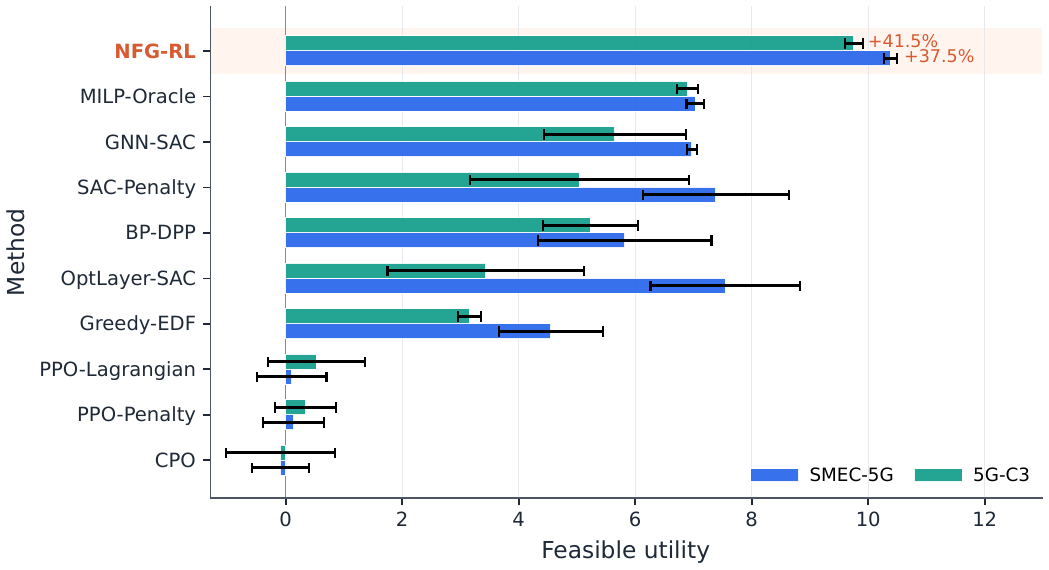}
    \vspace{-1.5mm}

    {\scriptsize (a) Public-trace comparison.}
\end{minipage}\hfill
\begin{minipage}[t]{0.495\columnwidth}
    \centering
    \includegraphics[width=\linewidth]
    {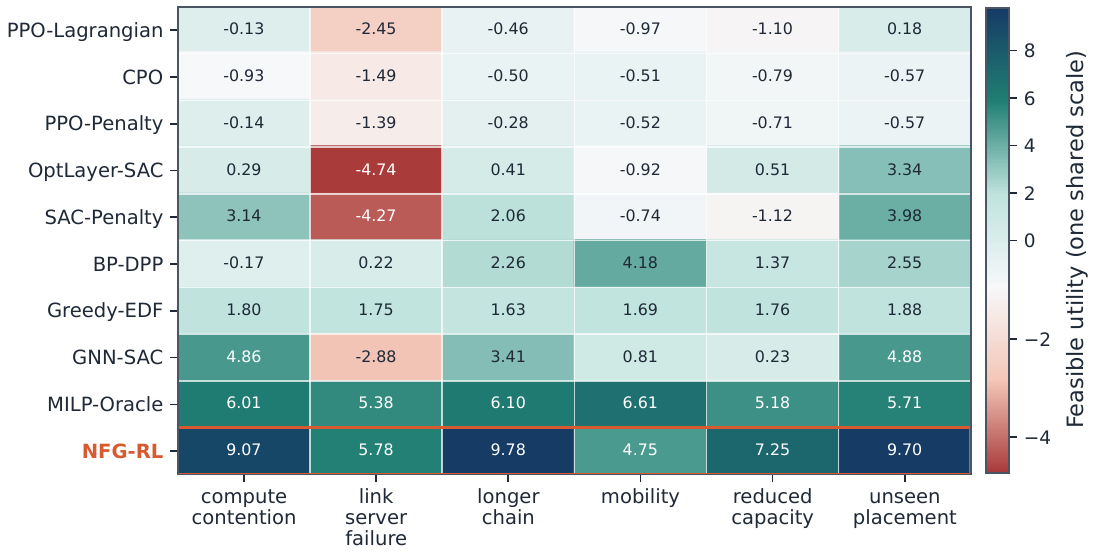}
    \vspace{-1.5mm}

    {\scriptsize (b) Six system shifts.}
\end{minipage}
\vspace{-1mm}
\caption{Public-trace-conditioned evaluation: (a) feasible utility in
SMEC-5G and 5G-C3 and (b) 5G-C3 utility under six perturbations, using a
shared color scale.}
\label{fig:real_results}
\vspace{-2mm}
\end{figure}

\noindent\textbf{Trace-conditioned results.}
In SMEC-5G, NFG-RL achieves utility $10.38$ versus $7.55$ for OptLayer-SAC ($37.5\%$ higher), reduces violation from $0.415$ to $0.163$ ($60.8\%$), and reduces P99 delay from $216.13$ to $92.96$ ($57.0\%$). In 5G-C3, it achieves $9.76$ versus $6.90$ for the myopic oracle ($41.5\%$ higher), reduces violation from $0.285$ to $0.146$ ($48.5\%$), and reduces P99 delay from $327.02$ to $80.04$ ($75.5\%$); it also provides the highest SLA satisfaction in both environments.

\subsubsection{Parametric 5G-C3 Stress Tests}

Without retraining, all policies are evaluated in 5G-C3 at $\rho=1.05$ for three episodes per seed under six deterministic shifts: higher mobility with channel quality $\times0.72$; compute contention with CPU $\times0.55$ and background load $+0.35$; a longer-chain proxy with demand $\times1.12$ and CPU $\times0.76$; reduced capacity with channel/CPU $\times0.66/0.68$; link/server failure with factors $0.35/0.30$ and an additional channel/CPU multiplier of $0.70$; and placement-resource stress with memory/CPU $\times0.62/0.82$.

NFG-RL achieves the highest feasible utility in five shifts; under mobility, the myopic oracle scores $6.61$ versus $4.75$ for NFG-RL. Thus, the results support robustness to most tested resource perturbations but not uniform superiority under mobility shift or cross-trace transfer.

\section{Conclusion}
\label{sec:conclusion}
This paper presented NFG-RL, a framework in which network feasibility enters the policy itself. Admissibility is represented as a residual inclusion, heterogeneous network constraints compile into typed primitives, proto-actions are transported through a local variational operator, and gradients flow through the transport map, so one compiled triple $(\bphi_{\mathfrak{N}},\cK_{\mathfrak{N}},J_{\Phi})$ determines where the executed policy places probability mass, how exploration noise is shaped at active constraints, and how critic gradients reach the actor.
The reduced wireless-edge surrogate illustrates how routing/placement
proxies, wireless interference, queueing, compute allocation, slicing,
reliability, tail latency, and discrete decisions can be represented through
one typed residual interface.
The public measurements condition the exogenous demand, channel, background
load, and mobility sequences, while counterfactual actions and performance
metrics are generated by the common surrogate transition model.
Across SMEC-5G and 5G-C3, NFG-RL raises feasible utility by 37.5--41.5\% over the strongest non-NFG baseline in each environment, cuts raw-action violation by 48.5--60.8\%, and lowers P99 delay by 57.0--75.5\%, confirming that feasibility transport improves both executability and control quality.

\nocite{*}
\bibliographystyle{IEEEtran}
\bibliography{refs}

\end{document}